\documentclass[
  aps,
  physrev,
  reprint,
  amsmath,
  amssymb,
  superscriptaddress,
  longbibliography,
  nofootinbib,
  floatfix
]{revtex4-2}

\usepackage[utf8]{inputenc}
\usepackage[T1]{fontenc}

\usepackage{graphicx}
\usepackage{bm}
\usepackage{mathtools}
\usepackage{amsthm}
\usepackage{booktabs,multirow}

\usepackage{tikz}
\usetikzlibrary{arrows.meta,positioning,calc,fit,decorations.pathreplacing}

\usepackage{microtype}
\allowdisplaybreaks

\usepackage[colorlinks=true,linkcolor=blue,citecolor=blue,urlcolor=blue]{hyperref}
\hypersetup{
  pdftitle={Stable Qubit Readout and the Identifiability of Population Change},
  pdfauthor={Dongdong Zhang}
}

\theoremstyle{definition}
\newtheorem{definition}{Definition}

\theoremstyle{plain}
\newtheorem{theorem}{Theorem}
\newtheorem{proposition}[theorem]{Proposition}
\newtheorem{corollary}[theorem]{Corollary}

\newcommand{\Tr}{\operatorname{Tr}}
\newcommand{\rank}{\operatorname{rank}}

\newcommand{\range}{\operatorname{range}}

\newcommand{\id}{\mathbb I}
\newcommand{\R}{\mathbb R}

\newcommand{\Herm}{\mathsf{Herm}}

\newcommand{\norminf}[1]{\left\lVert #1\right\rVert_\infty}
\newcommand{\normtwo}[1]{\left\lVert #1\right\rVert_2}
\newcommand{\normone}[1]{\left\lVert #1\right\rVert_1}

\newcommand{\ket}[1]{\lvert #1\rangle}

\newcommand{\ketbra}[2]{\lvert #1\rangle\!\langle #2\rvert}

\newcommand{\calI}{\mathcal I}
\newcommand{\calD}{\mathcal D}
\newcommand{\Nz}{d_z}
\newcommand{\qzero}{\bm q_0}
\newcommand{\ez}{\bm e_z}

\begin{document}

\title{Stable Qubit Readout and the Identifiability of Population Change}
\author{Dongdong Zhang}
\email{don@hfit.edu.cn}
\affiliation{School of Aerospace Information, Hefei Institute of Technology, Hefei 238076, China}

\begin{abstract}
Stable readout statistics are often taken as evidence for a well-defined physical response, but stability alone need not identify which state quantity has changed. We analyze this issue for finite collections of qubit states measured by binary readouts, focusing on changes in computational-basis population. The central question is when reproducible response data certify the sign or range of an underlying population change. We show that the answer is controlled by the calibrated measurement directions, not by loop consistency alone. For a fully calibrated finite readout family, we derive an exact closed-form interval of all compatible population changes. We also construct a same-record, jointly measurable example in which identical probabilities and accepted loop checks admit positive, zero, and negative population interpretations. When only a diagonal readout gain and a bound on coherence sensitivity are trusted, we obtain the sharp minimax interval and the necessary-and-sufficient sign condition $g>2\chi$. These results separate implementation stability from population identifiability and provide analytic benchmarks for qubit readout calibration.
\end{abstract}

\maketitle

\section{Introduction}
\label{sec:intro}

Quantum experiments infer physical properties from outcome statistics produced by imperfect sources, controls, and detectors.  This readout-mediated structure appears in quantum-sensor searches for weak fundamental interactions, where a controlled phase can be encoded into a population readout~\cite{Rong2018NVSensor}.  The mathematical language of effects and positive-operator-valued measures (POVMs) makes the dependence on the measurement explicit~\cite{Busch1996,Holevo2011,NielsenChuang2010}.  Quantum state and detector tomography then ask how much of a state or measurement can be reconstructed from calibrated probe--response data or from an available measurement record~\cite{ParisRehacek2004,Hradil1997,James2001,Fiurasek2001,DAriano2003,Teo2013}.  Informational completeness is sufficient for full reconstruction, and symmetric or tight informationally complete designs provide standard routes to this goal~\cite{Renes2004,Scott2006}.  Many scientific questions, however, concern one target quantity or functional rather than the whole state or detector.  Partial and indirect estimation therefore ask which observables or functions are fixed by the actually available, possibly incomplete, statistics~\cite{DAriano2008,Carmeli2017,Keith2018}, while confidence-region methods attach finite-sample guarantees to tomography-based inferences~\cite{Christandl2012,Zambrano2024}.

A complementary literature addresses a different failure mode: the nominal source and measurement labels may not denote fixed physical operations.  Self-consistent tomography, loop tests, gate-set tomography, context-dependence diagnostics, and quantum characterization, verification, and validation (QCVV) expose state-preparation-and-measurement (SPAM) correlations, drift, gauge freedom, and unmodeled error~\cite{Merkel2013,VanEnk2013,Stark2014,Jackson2015,Greenbaum2015,McCormick2017,Rudinger2019,DiMatteo2020,BlumeKohout2020,Nielsen2021,Rudinger2022,Cattaneo2023,Jayakumar2024,BlumeKohout2025,Hashim2025}.  Other approaches use additional assumptions or repeated-measurement resources to separate state-preparation and measurement errors more directly~\cite{Lin2021,Khan2026}.  These tools test whether a specified operational model is self-consistent or context-independent, diagnose SPAM-related failure modes, and identify extra resources that can remove particular measurement ambiguities.  Such checks leave a further target-identifiability question: whether the calibrated readout directions and the observed finite response record determine a specified state quantity, such as a computational-basis population change.

Readout is where target identifiability and implementation stability meet.  Computational-basis calibration fixes the diagonal entries of a binary qubit effect, but a general effect can also contain off-diagonal components.  Detector tomography and direct measurement-characterization studies have reconstructed, selectively accessed, or operationally quantified such coherent readout response~\cite{Lundeen2009,Feito2009,Zhang2012,Chen2019,Cimini2019,Zhang2020,Xu2020,XuPRL2021,XuAP2021}.  Related measurement-device-independent random-number-generation experiments use real-time measurement tomography to prevent device-characterization errors from being promoted into certified physical claims~\cite{Nie2016MDIQRNG}.  Theory and numerical work has also analyzed coherent measurement noise and procedures for detecting or suppressing it~\cite{Tang2024,Malik:2025kpb}.  The same off-diagonal degrees of freedom can be quantified as measurement coherence; in particular, for a binary qubit POVM the $\ell_\infty$ measurement-coherence monotone reduces to the total off-diagonal magnitude of the two POVM effects~\cite{Baek2020}.  Classical assignment-matrix mitigation assumes a population-only response.  Detector tomography can test this assumption, detector-tomography-assisted classical mitigation is exact under invertible classical noise~\cite{Maciejewski2020}, and scalable assignment-matrix methods address correlated classical measurement errors within that classical-response model~\cite{Nation2021}.  Methods that explicitly retain coherent response go beyond the assignment-matrix setting~\cite{Aasen2024,Malik:2025kpb}.  Because full detector characterization can be resource intensive and finite-precision measurement tomography is itself an active subject~\cite{Barbera2025,ZambranoQMT2025,Das2026}, it is useful to isolate what different levels of readout calibration certify about a chosen state quantity.

The general span criterion for target identifiability is already well established, and so is the observation that off-diagonal POVM elements couple to state coherence.  Here we specialize this general issue to a finite-dimensional readout benchmark: given a stable finite-family response, determine the exact set of computational-basis population changes compatible with the same calibrated readout record, and identify what additional readout information certifies the sign of that change.  The question is timely because bounded or even very small measurement deviations can materially alter quantum certification~\cite{Rosset2012,Morelli2022,Tavakoli2024,Svegborn2025,Zambrano2026,Moreno2026}.  Closely related repeated-measurement work shows that physicality constraints can turn an otherwise gauge-like measurement ambiguity into a finite compatible interval~\cite{Chu2026}.  We focus on binary qubit readouts, where the relevant target-identifiability problem can be solved in closed form using Bloch-ball geometry.  This gives an analytic separation between implementation stability, understood as reproducibility of a finite operational response, and population identifiability, understood as determination of a computational-basis population change by calibrated measurement directions.

Within this scope, the paper makes the following contributions.

First, for an arbitrary fully calibrated finite family of binary qubit effects, we reduce the signed readout differences to a real linear system $\bm g=A\bm q$ on the Bloch unit ball for state differences and derive the exact compatible interval for the computational-basis population change.  Its width
\begin{equation}
W_z=2\left\lVert P_{\ker A}\ez\right\rVert_2
\sqrt{1-\normtwo{A^+\bm g}^{2}}
\label{eq:intro-width}
\end{equation}
factorizes into the overlap between the population direction and the measurement null space, and the data-dependent physical slack left by the observed response.  Thus the result turns the established operator-span and set-identifiability criteria into an attainable closed-form interval for this qubit target.

Second, we give an explicit same-record counterexample.  The two binary effects are coarse grainings of a single parent POVM, so both reported readout bits can be computed from the same underlying outcome record.  Three valid triples of qubit states reproduce exactly the same six probabilities and the same computational-basis calibration, and they pass the same specified loop test, yet they imply positive, zero, and negative population shifts.  The construction identifies a common-mode coherent response that is invisible to this implementation comparison.

Third, we derive the exact compatible population interval when the readout effect is not fully reconstructed but the diagonal gain $\kappa$ and an independent bound $|c|\le\chi$ on the off-diagonal response are trusted.  For an observed scalar response $g$, the lower endpoint is
\begin{equation}
z_-=
\frac{\kappa g-2\chi\sqrt{\kappa^2+4\chi^2-g^2}}
{\kappa^2+4\chi^2},
\label{eq:intro-zminus}
\end{equation}
for compatible data, and the population sign is identified if and only if $g>2\chi$.  The threshold has a direct resource interpretation: the observed response must exceed the largest contribution that the allowed measurement coherence can generate without any population change.  The formula provides an analytic qubit benchmark for numerical semidefinite-programming and confidence-region analyses when only partial readout calibration is trusted.

Finally, we use these results to organize three calibration regimes: full calibration of the binary effects, diagonal calibration supplemented by a scalar bound on coherent response, and interval-valued calibration reduced to second-order-cone constraints.  In all three regimes, finite-family loop data and readout-calibration data play distinct roles: the former audit implementation stability, while the latter determine which measurement directions are available for population identifiability.

This work is organized as follows.  Section~\ref{sec:architecture} defines the operational and physical layers used throughout the analysis.  Section~\ref{sec:exact-interval} derives the exact population interval for a fully calibrated finite readout family.  Section~\ref{sec:benchmark} presents the same-record ambiguity construction.  Section~\ref{sec:bounded} gives the bounded-coherence population certificate.  Section~\ref{sec:closure} extends the analysis to calibration intervals and finite-sample guarantees.  Section~\ref{sec:discussion} places the results in relation to tomography, QCVV, measurement coherence, and imperfect-measurement certification.

\section{Two inference layers in a finite-family experiment}
\label{sec:architecture}

\subsection{Registered source settings and implementation loop}

Let $i\in\{0,1\}$ label two registered active source settings, and let $-$ label one registered reference setting common to both active settings.  Let $j\in\{0,1\}$ label two reported binary readouts.  The two readouts may be deterministic or randomized functions of the same underlying outcome record; no within-shot independence assumption is imposed.  We write
\begin{equation}
 p_{ij}^{+}:=P(Y_j=1\mid a_i^+),
 \qquad
 p_j^-:=P(Y_j=1\mid a^-),
 \label{eq:prob-def}
\end{equation}
and define the active--reference response
\begin{equation}
 g_{ij}:=p_{ij}^{+}-p_j^-.
 \label{eq:gij}
\end{equation}
The active and reference settings, the reported readout maps, the convention for the outcome labeled $1$, and the relation between the reported readouts and the underlying record are fixed before these response data are given a physical interpretation.

For the implementation-loop audit, introduce the endpoint pairs
\begin{equation}
 x_{ij}:=(p_{ij}^{+},p_j^-),
\end{equation}
so that each pair contains the active and reference endpoint probabilities for the same reported readout $j$.  Let
\begin{equation}
 \bm p=(p_{00}^{+},p_{10}^{+},p_{01}^{+},p_{11}^{+},p_0^-,p_1^-)^{\mathsf T}
\end{equation}
denote the vector of all six endpoint probabilities.  When finite-sample uncertainty is considered, $x_{ij}=x_{ij}(\bm p)$ and $g_{ij}=g_{ij}(\bm p)$ are understood as functions of this vector.

The vertical direction in the implementation square corresponds to replacing readout $0$ by readout $1$.  This readout replacement is allowed to carry a common scalar offset: if two readout effects differ only by an identity component, $F_1=F_0+\Delta\id$, then every state probability is shifted by the same amount $\Delta$, while active--reference differences are unchanged.  We therefore define
\begin{equation}
 \mathcal A_{\Delta}(u,v)=(u+\Delta,v+\Delta),
 \qquad \Delta\in\calD,
 \label{eq:translation}
\end{equation}
where $\calD$ is the independently calibrated set of admissible relative offsets between the two readouts.  Only the relative offset matters; equivalently, separate readout offsets $\delta_0$ and $\delta_1$ enter the loop only through $\Delta=\delta_1-\delta_0$.  The value of $\Delta$ is not fitted separately for each edge or each active setting.

For a two-component vector $y=(y_1,y_2)$, write
\begin{equation}
 \norminf{y}:=\max\{|y_1|,|y_2|\}.
\end{equation}
The horizontal residual compares the two active settings at fixed readout $j$.  Since
\[
 x_{1j}-x_{0j}
 =
 (p_{1j}^{+}-p_{0j}^{+},\,0),
\]
we have
\begin{align}
 r_j^{\rm h}(\bm p)
 &=\norminf{x_{1j}(\bm p)-x_{0j}(\bm p)} \nonumber\\
 &=|p_{1j}^{+}-p_{0j}^{+}|.
 \label{eq:rh}
\end{align}
Thus $r_j^{\rm h}$ is the mismatch along the source-replacement edge for readout $j$.

The vertical residual compares the two readouts at fixed active setting $i$, after the common scalar translation $\Delta$ is applied to both coordinates of the readout-$0$ endpoint pair.  Since
\[
 x_{i1}-\mathcal A_\Delta x_{i0}
 =
 (p_{i1}^{+}-p_{i0}^{+}-\Delta,\,
  p_1^{-}-p_0^{-}-\Delta),
\]
we have
\begin{align}
 r_i^{\rm v}(\Delta;\bm p)
 &=\norminf{x_{i1}(\bm p)-\mathcal A_{\Delta}x_{i0}(\bm p)} \nonumber\\
 &=
 \max\left\{
 |p_{i1}^{+}-p_{i0}^{+}-\Delta|,
 |p_1^{-}-p_0^{-}-\Delta|
 \right\}.
 \label{eq:rv}
\end{align}
Thus $r_i^{\rm v}$ asks whether replacing readout $0$ by readout $1$ looks like the same scalar offset $\Delta$ on both the active endpoint and the reference endpoint.

The loop audit asks whether the two horizontal residuals $(j=0,1)$ and the two vertical residuals $(i=0,1)$ are all small for one common value of $\Delta$ allowed by the independently calibrated set $\calD$.  A simultaneous confidence region $\calI\subset[0,1]^6$ for $\bm p$ then gives finite-sample lower bounds on the responses $g_{ij}$ and finite-sample upper bounds on the edge residuals.  The vertical part is evaluated with the prescribed common-$\Delta$ optimization over $\Delta\in\calD$.

\begin{definition}[Finite-family implementation certificate]
\label{def:implementation}
Fix a response margin $\gamma>0$ and an edge tolerance $\tau>0$ before inspecting the response data.  The margin $\gamma$ is the required lower bound on every active--reference response, and $\tau$ is the allowed mismatch on each edge of the implementation loop.  Let $\calI$ be the simultaneous confidence region for the endpoint-probability vector $\bm p$, and let $\calD$ be the independently calibrated set of admissible scalar translations.

Define the worst-case certified response margin by
\begin{equation}
 M_{\rm resp}:=
 \min_{i,j}\inf_{\bm p\in\calI}g_{ij}(\bm p).
\end{equation}
Define the horizontal loop error by
\begin{equation}
 H_{\rm loop}:=
 \max_j\sup_{\bm p\in\calI}r_j^{\rm h}(\bm p),
\end{equation}
and define the vertical loop error by
\begin{equation}
 V_{\rm loop}:=
 \inf_{\Delta\in\calD}
 \sup_{\bm p\in\calI}
 \max_i r_i^{\rm v}(\Delta;\bm p).
\end{equation}
Here the same $\Delta$ must be used for both vertical edges; it is chosen from the independently calibrated set $\calD$ and is not fitted separately for the two active settings.

The registered family passes the implementation certificate when
\begin{equation}
 M_{\rm resp}\ge\gamma
 \label{eq:margin-cert}
\end{equation}
and
\begin{equation}
 E_{\rm loop}:=\max\{H_{\rm loop},V_{\rm loop}\}\le\tau.
 \label{eq:loop-cert}
\end{equation}
\end{definition}

The two conditions have different roles.  Equation~\eqref{eq:margin-cert} says that every response remains positive by at least the predeclared margin throughout the simultaneous confidence region.  Equation~\eqref{eq:loop-cert} says that the registered implementation square closes within tolerance: the two horizontal edges are small, and the two vertical edges are small for one common translation $\Delta$ allowed by the independent calibration.

The loop condition is a simultaneous four-edge claim.  The closing edge shares endpoint data with the other edges and is therefore statistically correlated with them, but this statistical dependence does not make it logically redundant at the same tolerance.  Appendix~\ref{app:loop} proves the sharp deterministic statement: three consecutive residuals bounded by $\epsilon$ imply, in general, only a $3\epsilon$ bound on the fourth.

\begin{figure*}[t]
\centering
\begin{minipage}[t]{0.47\textwidth}
\centering
\begin{tikzpicture}[x=1.52cm,y=1.20cm,>=Latex,
  node/.style={draw,circle,inner sep=1.7pt,font=\small},
  lab/.style={font=\scriptsize,fill=white,inner sep=1pt}]
\node[node] (x00) at (0,1) {$x_{00}$};
\node[node] (x10) at (2,1) {$x_{10}$};
\node[node] (x01) at (0,0) {$x_{01}$};
\node[node] (x11) at (2,0) {$x_{11}$};
\draw[<->,thick] (x00)--node[lab,above] {source replacement}(x10);
\draw[<->,thick] (x01)--node[lab,below] {source replacement}(x11);
\draw[<->,thick] (x00)--node[lab,left] {readout replacement}(x01);
\draw[<->,thick] (x10)--node[lab,right] {readout replacement}(x11);
\node[font=\scriptsize,align=center] at (1,-0.63)
{all four edges are certified jointly};
\end{tikzpicture}
\par\smallskip
{\small (a) Registered implementation square.}
\end{minipage}\hfill
\begin{minipage}[t]{0.50\textwidth}
\centering
\begin{tikzpicture}[>=Latex,node distance=0.31cm,
 box/.style={draw,rounded corners,align=center,font=\scriptsize,text width=5.15cm,minimum height=0.69cm},
 gate/.style={draw,rounded corners,align=center,font=\scriptsize,text width=4.65cm,minimum height=0.67cm}]
\node[box] (data) {registered endpoint probabilities\\and active--reference responses};
\node[box,below=of data] (loop) {implementation audit\\margin and full-loop consistency};
\node[gate,below=0.47cm of loop] (model) {measurement-information layer\\full readout calibration, a coherence bound, or a robust set};
\node[box,below=of model] (physics) {physical target certificate\\exact interval or one-sided bound on $\Delta z$};
\draw[->,thick] (data)--(loop);
\draw[->,thick] (loop)--(model);
\draw[->,thick] (model)--(physics);
\draw[->,dashed,bend left=42] (loop.east) to node[right,font=\scriptsize,align=left] {not implied by\\loop closure} (physics.east);
\end{tikzpicture}
\par\smallskip
{\small (b) Logical architecture.}
\end{minipage}
\caption{Implementation stability and target identification require different evidence.  The loop audits the declared substitutions.  Additional measurement information determines whether a stable response yields a certificate for the population change.}
\label{fig:architecture}
\end{figure*}

\subsection{Qubit response coordinates}
\label{subsec:qubit-coordinates}

Once a qubit description is adopted, write the effect associated with the reported outcome $Y_j=1$ in the computational basis as
\begin{equation}
 F_j=
 \begin{pmatrix}
 \beta_j & c_j\\
 c_j^* & \beta_j+\kappa_j
 \end{pmatrix},
 \qquad 0\le F_j\le\id.
 \label{eq:effect-param}
\end{equation}
The registered outcome label is taken, as part of the calibration convention, so that the calibrated diagonal gain $\kappa_j$ is nonnegative.  For an active--reference state difference
\begin{equation}
 X=\rho-\sigma=
 \begin{pmatrix}
 -z & w^*\\
 w & z
 \end{pmatrix},
 \qquad w=u+iv,
 \label{eq:X-param}
\end{equation}
define
\begin{equation}
 \bm q=(z,u,v)^{\mathsf T}.
\end{equation}
Then the response of readout $j$ to the state difference is
\begin{equation}
 \Tr(F_jX)
 =
 \kappa_j z+2\operatorname{Re}(c_jw).
\end{equation}
Thus the population coordinate $z$ and the coherence coordinates $(u,v)$ enter the same scalar response unless the effect is diagonal in the computational basis.

The eigenvalues of $X$ are $\pm\normtwo{\bm q}$, and $X$ is a difference of two qubit states if and only if
\begin{equation}
 \normtwo{\bm q}\le1.
 \label{eq:difference-ball}
\end{equation}
Sufficiency follows from $\rho=(\id+X)/2$ and $\sigma=(\id-X)/2$; Appendix~\ref{app:geometry} records the details.

For active source setting $i$, let $X_i=\rho_i-\sigma$ and write
\begin{equation}
 g_{ij}=\Tr(F_jX_i)
 =\kappa_j z_i+2\operatorname{Re}(c_jw_i).
 \label{eq:response}
\end{equation}
Writing $c_j=a_j+ib_j$, define the real row vector
\begin{equation}
 \bm a_j=(\kappa_j,2a_j,-2b_j).
 \label{eq:row}
\end{equation}
More generally, for a calibrated family of $m$ binary effects, stack the rows into $A\in\R^{m\times3}$.  For each active setting $i$, collect the responses into
$\bm g_i=(g_{i1},\ldots,g_{im})^{\mathsf T}$ and write
\begin{equation}
 \bm g_i=A\bm q_i,
 \qquad \normtwo{\bm q_i}\le1.
 \label{eq:linear-ball}
\end{equation}
The target population change for active setting $i$ is the first coordinate,
\begin{equation}
 z_i=\ez^{\mathsf T}\bm q_i,
 \qquad \ez=(1,0,0)^{\mathsf T}.
\end{equation}
When analyzing a single active--reference difference, we suppress the index $i$ and write $\bm g=A\bm q$ and $z=\ez^{\mathsf T}\bm q$.

\subsection{Operator-span identifiability background}
\label{subsec:span}

For a general finite-dimensional system, let
\begin{equation}
\widetilde F_j=F_j-\frac{\Tr F_j}{d}\id
\end{equation}
denote the traceless part of effect $F_j$.  A trace-zero target observable $T$ is determined, for all state differences, by the effects $F_1,\ldots,F_m$ exactly when $T$ belongs to the real span of $\widetilde F_1,\ldots,\widetilde F_m$~\cite{DAriano2008,Carmeli2017,Keith2018}.  Equivalently, every traceless Hermitian operator in the measurement kernel is Hilbert--Schmidt orthogonal to $T$.  We use this established criterion as the identifiability background.  The results below specialize the compatible set to the qubit difference ball and compute the resulting population interval in closed form.

In the qubit coordinates above, exact population identifiability for all compatible state differences is therefore equivalent to
\begin{equation}
\ez\in\operatorname{row}(A),
\label{eq:row-identifiability}
\end{equation}
where $\operatorname{row}(A)$ denotes the real row space of $A$, i.e., the span of the row vectors $\bm a_j$.  Equivalently, if $\ker A:=\{\bm h\in\R^3:A\bm h=0\}$ denotes the null space of $A$, then Eq.~\eqref{eq:row-identifiability} is equivalent to $P_{\ker A}\ez=0$, where $P_{\ker A}$ denotes the Euclidean orthogonal projector onto $\ker A$.  The next section computes the full compatible population interval in closed form.  When Eq.~\eqref{eq:row-identifiability} fails, the interval has nonzero width, and its two factors separate the measurement-direction blind spot from the data-dependent physical slack.

A particularly important case is a scalar-shift readout family,
\begin{equation}
F_j=F_\star+\delta_j\id,
\label{eq:scalar-shift}
\end{equation}
with each $F_j$ still satisfying $0\le F_j\le\id$.  All effects then have the same traceless part and hence the same response row $\bm a_\star$.  Since state differences are traceless, the scalar offsets do not change $\Tr(F_jX)$ and cannot increase $\rank(A)$.  They may be valuable for implementation checks, but they do not add measurement-direction information.  In particular, any off-diagonal component common to $F_\star$ can affect the response, but it is not diagnosed by comparisons among the scalar-shifted implementations.

\section{Exact population ambiguity for calibrated readout families}
\label{sec:exact-interval}

For the rest of this section, fix one active--reference difference and suppress the active-setting index.  The physical inverse problem in Eq.~\eqref{eq:linear-ball} is the intersection of an affine space with the Bloch unit ball for state differences.  For a qubit, this geometry gives an exact interval for the target population change.

\begin{theorem}[Exact calibrated-family population interval]
\label{thm:exact-interval}
Let $A\in\R^{m\times3}$ be a fully calibrated response matrix and let $\bm g\in\range(A)$ be an exactly consistent response vector.  Let $A^+$ denote the Moore--Penrose pseudoinverse, and define the minimum-norm solution
\begin{equation}
\qzero=A^+\bm g
\label{eq:q0}
\end{equation}
and the population blind-spot factor
\begin{equation}
\Nz:=\normtwo{P_{\ker A}\ez}.
\label{eq:dz-factor}
\end{equation}
A compatible physical state difference exists if and only if $\normtwo{\qzero}\le1$.  When it exists, the exact set of compatible population changes is
\begin{equation}
\boxed{
z\in\left[
\ez^{\mathsf T}\qzero-\Nz\sqrt{1-\normtwo{\qzero}^{2}},
\ \ez^{\mathsf T}\qzero+\Nz\sqrt{1-\normtwo{\qzero}^{2}}
\right].}
\label{eq:exact-interval}
\end{equation}
Both endpoints are attainable by physical pairs of qubit states.
\end{theorem}

\begin{proof}
Every solution of $A\bm q=\bm g$ has the orthogonal decomposition
\begin{equation}
\bm q=\qzero+\bm n,
\qquad \bm n\in\ker A,
\qquad \qzero\perp\bm n,
\end{equation}
because $\qzero=A^+\bm g$ is the minimum-norm solution.  Hence
\begin{equation}
\normtwo{\bm q}^{2}
=
\normtwo{\qzero}^{2}+\normtwo{\bm n}^{2}.
\end{equation}
A compatible physical state difference exists if and only if $\normtwo{\qzero}\le1$.  In that case the physical condition is
\begin{equation}
\normtwo{\bm n}\le r,
\qquad r:=\sqrt{1-\normtwo{\qzero}^{2}}.
\end{equation}
The target coordinate is
\begin{equation}
\begin{aligned}
 \ez^{\mathsf T}\bm q
 &=
 \ez^{\mathsf T}\qzero+\ez^{\mathsf T}\bm n \\
 &=
 \ez^{\mathsf T}\qzero+(P_{\ker A}\ez)^{\mathsf T}\bm n .
\end{aligned}
\end{equation}
Thus only the component of $\ez$ in $\ker A$ can vary on the compatible ball.  Cauchy--Schwarz gives
\begin{equation}
\min_{\substack{\bm n\in\ker A\\\normtwo{\bm n}\le r}}
\ez^{\mathsf T}\bm n
=
-r\normtwo{P_{\ker A}\ez},
\end{equation}
and the maximum is $r\normtwo{P_{\ker A}\ez}$.  If $r\normtwo{P_{\ker A}\ez}=0$, the interval is degenerate and the endpoint is reached by $\bm n=0$.  Otherwise the two endpoints are reached by choosing $\bm n$ parallel or antiparallel to $P_{\ker A}\ez$ with norm $r$.  Equation~\eqref{eq:difference-ball} then converts the endpoint vectors into physical pairs of qubit states.
\end{proof}

When a compatible physical state difference exists, the interval width is
\begin{equation}
W_z=2\Nz\sqrt{1-\normtwo{\qzero}^{2}}.
\label{eq:width}
\end{equation}
The first factor, $\Nz\in[0,1]$, depends only on the calibrated measurement directions.  It vanishes exactly when the population axis belongs to the response row space, so that the population change is identified for all compatible state differences.  The second factor is data dependent: it is the remaining radius of the compatible slice of the Bloch difference ball. Thus responses near the boundary of the reachable response set
$\{A\bm q:\bm q\in\R^3,\ \normtwo{\bm q}\le1\}$
leave little state-space freedom even when the measurement family is not population-identifying in the operator-span sense.

\begin{corollary}[Universal and data-specific identification]
\label{cor:identification}
The population change is identifiable for every physically compatible response vector if and only if $\Nz=0$, equivalently $\ez\in\operatorname{row}(A)$.  For a fixed physically compatible response vector, the compatible interval collapses if and only if
\begin{equation}
\Nz=0
\qquad\text{or}\qquad
\normtwo{\qzero}=1.
\end{equation}
Thus a measurement family may fail universal population identification, $\Nz>0$, while a particular response still identifies the population change because it lies on the boundary of the physical response image.
\end{corollary}

The second case is boundary-induced uniqueness, not row-space identifiability, let alone informational completeness.  It can be fragile under finite statistics because any enlargement of the response region generally restores nonzero state-space slack.

For one calibrated effect with nonzero response row, $A$ has one row
$\bm a=(\kappa,2\operatorname{Re}c,-2\operatorname{Im}c)$.  The directional factor becomes
\begin{equation}
\Nz=\frac{2|c|}{\sqrt{\kappa^2+4|c|^2}}.
\label{eq:single-dz}
\end{equation}
Thus a diagonal effect with nonzero diagonal gain has no directional blind spot for population, whereas coherent response rotates the measured direction away from the population axis.  A family related only by scalar offsets has exactly the same response row, and hence the same $\Nz$, as one member.  By contrast, adding calibrated readout directions cannot increase $\Nz$ and can reduce it when the new directions constrain the previous blind component; two suitable independent directions can make $\Nz$ vanish.

\subsection{Single-state population interval}
\label{subsec:absolute-state}

For a single qubit state
\begin{equation}
 \rho(z,w)=
 \begin{pmatrix}
 1-z&w^*\\w&z
 \end{pmatrix},
\end{equation}
define the centered vector
\begin{equation}
\bar{\bm q}=(z-1/2,\operatorname{Re}w,\operatorname{Im}w)^{\mathsf T}.
\end{equation}
Positivity of $\rho$ is equivalent to $\normtwo{\bar{\bm q}}\le1/2$.  If $p_j=\Tr(F_j\rho)$, then
\begin{equation}
p_j-\left(\beta_j+\frac{\kappa_j}{2}\right)
=\bm a_j^{\mathsf T}\bar{\bm q}.
\label{eq:centered-prob}
\end{equation}
Thus the same affine-ball geometry as in Theorem~\ref{thm:exact-interval} applies, with the unit ball for state differences replaced by the radius-$1/2$ ball for centered single-state coordinates.

\begin{corollary}[Exact calibrated-state population interval]
\label{cor:absolute}
Let $\bar{\bm p}\in\R^m$ be the centered probability vector with components
\begin{equation}
\bar p_j=p_j-\left(\beta_j+\frac{\kappa_j}{2}\right),
\end{equation}
and assume $\bar{\bm p}\in\range(A)$.  Let
\begin{equation}
\bar{\bm q}_0=A^+\bar{\bm p}.
\end{equation}
A compatible qubit state exists if and only if $\normtwo{\bar{\bm q}_0}\le1/2$.  When it exists, define
\begin{equation}
\bar z_0=\frac12+\ez^{\mathsf T}\bar{\bm q}_0,
\qquad
r_z=\Nz\sqrt{\frac14-\normtwo{\bar{\bm q}_0}^{2}}.
\end{equation}
The exact compatible state population interval is
\begin{equation}
\boxed{z\in[\bar z_0-r_z,\bar z_0+r_z].}
\label{eq:absolute-interval}
\end{equation}
Both endpoints are attainable by physical qubit states.
\end{corollary}

When the active and reference preparations are constrained only by their own absolute readout data, the feasible set for the pair of states is the Cartesian product of the two single-state feasible sets.  If the corresponding exact population intervals are $[L_+,U_+]$ and $[L_-,U_-]$, then the exact attainable interval for the population difference is their interval difference,
\begin{equation}
 \Delta z\in[L_+-U_-,\,U_+-L_-].
\end{equation}
This statewise construction can be tighter than using only the response difference together with the radius-one ball for $X=\rho-\sigma$, because the latter discards information contained in the two absolute response records.

\begin{figure}[t]
\centering
\begin{tikzpicture}[scale=1.75,>=Latex,line cap=round,line join=round]
\coordinate (O) at (0,0);
\coordinate (A) at (-0.361,-0.932);
\coordinate (B) at (0.721,0.692);
\coordinate (Q0) at (0.180,-0.120);
\coordinate (Az) at (-0.361,0.085);
\coordinate (Bz) at (0.721,0.085);

\fill[black!4] (O) circle (1);
\draw[thick] (O) circle (1);
\draw[->] (-1.18,0)--(1.22,0) node[right,font=\scriptsize] {$z$};
\draw[->] (0,-1.18)--(0,1.22) node[above,font=\scriptsize] {$n$};

\draw[semithick,black!45] (-0.55,-1.22)--(0.91,0.98);
\draw[very thick] (A)--(B);
\fill (A) circle (0.9pt);
\fill (B) circle (0.9pt);

\draw[dashed,black!55] (O)--(Q0);
\fill (Q0) circle (1.35pt)
 node[below right,font=\scriptsize,inner sep=1pt] {$\qzero$};

\draw[densely dashed,black!40] (A)--(-0.361,0);
\draw[densely dashed,black!40] (B)--(0.721,0);
\draw[semithick,<->,black!70] (Az)--(Bz)
 node[midway,above=1.2pt,font=\scriptsize,fill=white,inner sep=1pt]
 {$z$ interval};
\draw[thin,black!55] (-0.361,0.02)--(-0.361,0.15);
\draw[thin,black!55] (0.721,0.02)--(0.721,0.15);

\node[font=\scriptsize,align=center,fill=white,inner sep=1.3pt,rounded corners=1pt]
 at (-0.62,0.72) {physical ball\\$\normtwo{\bm q}\le1$};
\node[font=\scriptsize,align=center,fill=white,inner sep=1.3pt,rounded corners=1pt]
 at (0.58,-0.70) {affine slice\\$A\bm q=\bm g$};
\end{tikzpicture}
\caption{Geometric content of Theorem~\ref{thm:exact-interval} in a two-dimensional section.  The response fixes an affine slice of the Bloch difference ball.  The thick chord is the compatible physical slice, and $\qzero$ is the minimum-norm point on that slice.  The population ambiguity is the projection of this slice onto the $z$ axis; its radius is the product of the null-space projection factor $\Nz$ and the residual ball radius.}
\label{fig:ball-slice}
\end{figure}
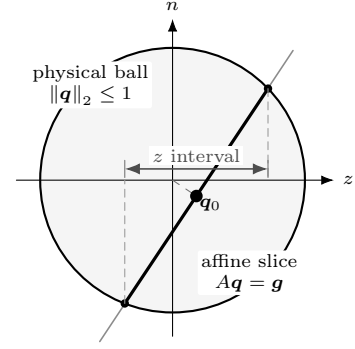

\section{A same-record population-sign ambiguity benchmark}
\label{sec:benchmark}

Theorem~\ref{thm:exact-interval} quantifies population ambiguity once the readout effects are fully known.  We now give an exact finite-family data table in which the registered implementation checks are satisfied, yet the sign of the compatible population change is not identified.  The construction rules out three possible explanations for the ambiguity: the two binary readout effects are legal POVM effects; they are coarse grainings of a single parent POVM, so both reported readouts can be computed from the same underlying outcome record; and all alternative qubit-state triples reproduce exactly the same calibrated probabilities.

\subsection{One parent measurement and two reported readouts}

Consider
\begin{equation}
\begin{aligned}
 F_0&=
 \begin{pmatrix}
 1/4&9/50\\[1mm]
 9/50&17/20
 \end{pmatrix},\\[1mm]
 F_1&=
 \begin{pmatrix}
 7/20&9/50+i/100\\[1mm]
 9/50-i/100&19/20
 \end{pmatrix}.
\end{aligned}
 \label{eq:benchmark-effects}
\end{equation}
Both satisfy $0\le F_j\le\id$.  Moreover, the operators
\begin{equation}
 F_0,
 \qquad F_1-F_0,
 \qquad \id-F_1
 \label{eq:parent-positive}
\end{equation}
are positive semidefinite.  Hence
\begin{equation}
 M_{11}=F_0,
 \quad M_{01}=F_1-F_0,
 \quad M_{00}=\id-F_1,
 \quad M_{10}=0
 \label{eq:parent-povm}
\end{equation}
defines a four-outcome parent POVM.  The two binary readouts are its coarse grainings,
\begin{equation}
 F_0=M_{10}+M_{11},
 \qquad
 F_1=M_{01}+M_{11}.
\end{equation}
This gives an explicit joint-measurability realization~\cite{Heinosaari:2008jrh,Yu2010,Guhne2023}: both reported binary outcomes can be computed from the same underlying outcome record, rather than from incompatible measurements.

The computational-basis calibration is
\begin{equation}
 (\beta_0,\kappa_0)=(1/4,3/5),
 \qquad
 (\beta_1,\kappa_1)=(7/20,3/5).
 \label{eq:benchmark-basis}
\end{equation}
The two effects have the same diagonal gain and the same real coherent component $9/50$.  Their difference contains a scalar offset, which drops out on trace-zero state differences, and a small imaginary coherent component $1/100$ that changes the response direction.  In the real coordinates of Eq.~\eqref{eq:row},
\begin{equation}
 A=
 \begin{pmatrix}
 3/5&9/25&0\\
 3/5&9/25&-1/50
 \end{pmatrix}.
 \label{eq:benchmark-A}
\end{equation}
The null space is
\[
 \ker A=\operatorname{span}\{(-3/5,1,0)^{\mathsf T}\},
\]
and therefore
\begin{equation}
 \Nz=\frac{3}{\sqrt{34}}\simeq0.5145.
 \label{eq:benchmark-dz}
\end{equation}
Thus even complete knowledge of both effects leaves a substantial component of the population axis in the response kernel.

\subsection{Three physical state models for one probability table}

Write a qubit state as
\begin{equation}
 \rho(z,w)=
 \begin{pmatrix}
 1-z&w^*\\w&z
 \end{pmatrix},
 \qquad |w|^2\le z(1-z).
 \label{eq:rhozw}
\end{equation}
Table~\ref{tab:benchmark-states} lists three exact triples of qubit states.  Within each model, the first row is the reference preparation and the other two rows are the registered active preparations.  The reference preparation is not assumed to be the same across the three models; the claim is that each complete triple reproduces the same calibrated probability table for the effects in Eq.~\eqref{eq:benchmark-effects}.  All nine matrices are positive definite density operators.

\begin{table*}[t]
\caption{Three exact state realizations of the same calibrated probability table.  The models $\mathcal M_+$, $\mathcal M_0$, and $\mathcal M_-$ have positive, zero, and negative active--reference population shifts, respectively.}
\label{tab:benchmark-states}
\centering
\begin{ruledtabular}
\begin{tabular}{c|cc|cc|cc}
&\multicolumn{2}{c|}{$\mathcal M_+$}&\multicolumn{2}{c|}{$\mathcal M_0$}&\multicolumn{2}{c}{$\mathcal M_-$}\\
preparation&$z$&$w$&$z$&$w$&$z$&$w$\\\hline
reference $-$&$1/10$&$-1/6+i/10$&$1/5$&$-1/3+i/10$&$1/4$&$-5/12+i/10$\\
active $0+$&$3/10$&$1/6-i/10$&$1/5$&$1/3-i/10$&$9/50$&$11/30-i/10$\\
active $1+$&$3/10$&$13/72-i/20$&$1/5$&$25/72-i/20$&$9/50$&$137/360-i/20$
\end{tabular}
\end{ruledtabular}
\end{table*}

For every one of the three state triples in Table~\ref{tab:benchmark-states}, substitution into Eq.~\eqref{eq:benchmark-effects} gives the same six probabilities:
\begin{equation}
\begin{array}{c|cc}
 &F_0&F_1\\\hline
 \rho_-&1/4&87/250\\
 \rho_0^+&49/100&74/125\\
 \rho_1^+&99/200&149/250.
\end{array}
\label{eq:benchmark-probabilities}
\end{equation}
Equivalently, the four endpoint pairs are
\begin{align}
 x_{00}&=(49/100,1/4),&
 x_{10}&=(99/200,1/4),\\
 x_{01}&=(74/125,87/250),&
 x_{11}&=(149/250,87/250).
 \label{eq:benchmark-endpoints}
\end{align}
For the independently fixed translation $\Delta=1/10$, the four residuals are
\begin{equation}
 \bigl(r_0^{\rm h},r_0^{\rm v}(1/10),r_1^{\rm h},r_1^{\rm v}(1/10)\bigr)
 =
\left(\frac1{200},\frac1{500},\frac1{250},\frac1{500}\right),
 \label{eq:benchmark-residuals}
\end{equation}
and the four active--reference responses are
\begin{equation}
 (g_{00},g_{10},g_{01},g_{11})
 =
 \left(\frac6{25},\frac{49}{200},\frac{61}{250},\frac{31}{125}\right).
 \label{eq:benchmark-responses}
\end{equation}
However, if $\Delta z_i:=z_i^+-z_-$ denotes the computational-basis population shift for active preparation $i$, then both active preparations have
\begin{equation}
 \Delta z_0=\Delta z_1=
 \begin{cases}
 1/5,&\mathcal M_+,\\
 0,&\mathcal M_0,\\
 -7/100,&\mathcal M_-.
 \end{cases}
 \label{eq:benchmark-signs}
\end{equation}
Thus the same calibrated probability table and the same accepted finite-family loop are compatible with positive, zero, and negative population shifts.

\begin{proposition}[Exact same-record sign ambiguity]
\label{prop:benchmark}
The parent POVM in Eq.~\eqref{eq:parent-povm}, the basis calibration in Eq.~\eqref{eq:benchmark-basis}, the six probabilities in Eq.~\eqref{eq:benchmark-probabilities}, and the loop residuals in Eq.~\eqref{eq:benchmark-residuals} are all fixed.  Nevertheless, there exist three valid qubit-state triples realizing these same registered data with the positive, zero, and negative population shifts shown in Eq.~\eqref{eq:benchmark-signs}.  The ambiguity is therefore not caused by incompatible measurements, different raw records, illegal effects, or approximate fitting.  It is a genuine target-identification ambiguity: neither access to the same underlying outcome record nor acceptance by the finite-family implementation loop identifies the sign of the population change.
\end{proposition}

The kernel mechanism is explicit.  The trace-zero Hermitian operator
\begin{equation}
 K=
 \begin{pmatrix}
 -1&-5/3\\
 -5/3&1
 \end{pmatrix}
 \label{eq:benchmark-K}
\end{equation}
satisfies
\begin{equation}
 \Tr(F_0K)=\Tr(F_1K)=0,
 \qquad
 \Tr\!\left[\left(\ketbra{1}{1}-\frac{\id}{2}\right)K\right]=1.
 \label{eq:kernel-properties}
\end{equation}
Thus $K$ is invisible to both calibrated readout probabilities but visible to the population target.  Moving a full-rank state by a sufficiently small amount along $K$ preserves both measured probabilities while changing its computational-basis population.  Appendix~\ref{app:benchmark} verifies all positivity and determinant conditions exactly.

\subsection{Quantifying the benchmark with the exact interval}

The ambiguity can be read directly from Corollary~\ref{cor:absolute}.  For the three probability rows in Eq.~\eqref{eq:benchmark-probabilities}, the minimum-norm centered state vectors are
\begin{align}
 \bar{\bm q}_{0}^{(-)}&=(-25/68,-15/68,1/10)^{\mathsf T},\\
 \bar{\bm q}_{0}^{(0+)}&=(-5/68,-3/68,-1/10)^{\mathsf T},\\
 \bar{\bm q}_{0}^{(1+)}&=(-55/816,-11/272,-1/20)^{\mathsf T}.
 \label{eq:benchmark-q0}
\end{align}
Here the subscript $0$ denotes the minimum-norm solution, while the superscript labels the preparation row.  The resulting statewise population intervals are evaluated from Eq.~\eqref{eq:absolute-interval} and displayed in Table~\ref{tab:benchmark-intervals}.

\begin{table}[t]
\caption{Statewise population intervals obtained from Eq.~\eqref{eq:absolute-interval} for the common probability table.  Numerical endpoints are rounded to six decimal places.}
\label{tab:benchmark-intervals}
\centering
\begin{ruledtabular}
\begin{tabular}{c@{\quad}c}
preparation&compatible interval for $z$\\\hline
reference $-$&$[0.010409,\ 0.254297]$\\
active $0+$&$[0.178311,\ 0.674630]$\\
active $1+$&$[0.179855,\ 0.685341]$
\end{tabular}
\end{ruledtabular}
\end{table}

Taking interval differences for the active and reference statewise intervals gives, to six decimal places,
\begin{align}
 \Delta z_0&\in[-0.075985,\ 0.664221],\\
 \Delta z_1&\in[-0.074442,\ 0.674932].
 \label{eq:benchmark-difference-intervals}
\end{align}
Both intervals cross zero and contain all three values in Eq.~\eqref{eq:benchmark-signs}.  The interval theorem thus does more than assert failure of a span condition: it quantifies how much population freedom remains after the full calibrated probability table has been used.

The example also shows why scalar offsets are not measurement-direction diversity.  Most of the difference between $F_0$ and $F_1$ is a scalar offset, which is useful for the implementation loop but does not change the response row.  The small imaginary coherent component changes the two rows of $A$ only slightly, whereas the population ambiguity is dominated by their large common tilt away from $\ez$.

\section{Analytic certificates from bounded coherent response}
\label{sec:bounded}

Full effect tomography is not always available, and it is not always necessary for a target-specific conclusion.  We therefore consider a reduced information model in which the diagonal response gain is calibrated, while the off-diagonal response in the computational basis is not reconstructed but is constrained by an independently obtained scalar bound.  This section derives the exact worst-case, or minimax, population interval for that information model and the sharp condition under which the sign of the population change is certified.

\subsection{Exact minimax interval}

Consider one binary qubit effect with trusted diagonal calibration,
\begin{equation}
 F=
 \begin{pmatrix}
 \beta&c\\
 c^*&\beta+\kappa
 \end{pmatrix},
 \qquad \kappa>0,
 \label{eq:bounded-effect}
\end{equation}
where the diagonal entries are assumed to be feasible, $0\le\beta\le1-\kappa$.  Suppose an independent calibration gives the scalar bound
\begin{equation}
 |c|\le\chi.
 \label{eq:chi-bound}
\end{equation}
For fixed $\beta$ and $\kappa$, positivity of both $F$ and $\id-F$ implies the physical ceiling
\begin{equation}
 |c|^2\le\chi_{\rm pos}^2(\beta,\kappa):=
 \min\{\beta(\beta+\kappa),
 (1-\beta)(1-\beta-\kappa)\}.
 \label{eq:chi-pos}
\end{equation}
Thus the effective bound entering the exact interval is $\min\{\chi,\chi_{\rm pos}\}$.  Equivalently, below we take $\chi$ to denote this effective bound and assume without loss of sharpness that $0\le\chi\le\chi_{\rm pos}(\beta,\kappa)$.  Equation~\eqref{eq:chi-pos} is a positivity ceiling; it should not be confused with an independent precision calibration of $c$.

\begin{theorem}[Exact bounded-coherence population interval]
\label{thm:bounded-interval}
Let $g\ge0$, $\kappa>0$, and $0\le\chi\le\chi_{\rm pos}$.  The union of all population changes compatible with
\begin{equation}
 g=\kappa z+2\operatorname{Re}(cw),
 \qquad |c|\le\chi,
 \qquad z^2+|w|^2\le1
 \label{eq:bounded-model}
\end{equation}
is nonempty if and only if
\begin{equation}
 g\le R:=\sqrt{\kappa^2+4\chi^2}.
 \label{eq:bounded-compatible}
\end{equation}
When nonempty, it is the closed interval
\begin{equation}
 \boxed{
 z\in[z_-,z_+],
 \qquad
 z_{\pm}=\frac{\kappa g\pm2\chi\sqrt{R^2-g^2}}{R^2}.}
 \label{eq:zpm}
\end{equation}
Every point in the interval, including both endpoints, is attainable by a legal qubit realization in the stated information model.
\end{theorem}

\begin{proof}
For fixed $z$, physicality of the state difference gives $|w|\le\sqrt{1-z^2}$.  Since $|c|\le\chi$, the coherent contribution can take any value in the interval
\[
 [-2\chi\sqrt{1-z^2},\,2\chi\sqrt{1-z^2}]
\]
by choosing the relative phase and magnitude of $c$ and $w$.  Therefore $z$ is feasible if and only if
\begin{equation}
 |g-\kappa z|\le2\chi\sqrt{1-z^2}.
 \label{eq:feasibility-ineq}
\end{equation}
Squaring this equivalent condition gives
\begin{equation}
 (\kappa^2+4\chi^2)z^2-2\kappa gz+g^2-4\chi^2\le0.
\end{equation}
The discriminant is nonnegative exactly when $g\le R$, and the two roots are the endpoints in Eq.~\eqref{eq:zpm}.  Hence the feasible set of $z$ values is precisely the closed interval $[z_-,z_+]$.

It remains only to note attainability.  For any $z$ satisfying Eq.~\eqref{eq:feasibility-ineq}, set $s=\sqrt{1-z^2}$.  If $s>0$, choose $w=s$ real and choose $c=(g-\kappa z)/(2s)$ real.  Then $|c|\le\chi$ by Eq.~\eqref{eq:feasibility-ineq}, and Eq.~\eqref{eq:bounded-model} is satisfied.  If $s=0$, Eq.~\eqref{eq:feasibility-ineq} forces $g=\kappa z$, and the response is attained with $w=0$ and any allowed $c$, for example $c=0$.  Since $z^2+|w|^2\le1$, Eq.~\eqref{eq:difference-ball} realizes the corresponding state difference by a physical pair of qubit states.  The condition $\chi\le\chi_{\rm pos}$ ensures that the chosen effect is legal.
\end{proof}

The lower endpoint $z_-$ is the sharp worst-case lower bound on the population change in this information model.  Thus $z_->0$ is the corresponding one-sided certificate for a positive population change.  When $\chi=0$, Eq.~\eqref{eq:zpm} collapses, for compatible data, to the diagonal-response identity $z=g/\kappa$.  For $\chi>0$, the interval accounts for the most adverse coherent contribution allowed by the independent bound $|c|\le\chi$: part of the observed response may be explained by state coherence rather than by population change.

\begin{corollary}[Sharp sign threshold]
\label{cor:sharp-sign}
For compatible data with $g\ge0$, every realization has $z>0$ if and only if
\begin{equation}
 \boxed{g>2\chi.}
 \label{eq:sharp-threshold}
\end{equation}
At $g=2\chi$, the lower endpoint satisfies $z_-=0$; for $0\le g<2\chi$, a zero-population realization is also feasible.
\end{corollary}

The criterion is independent of $\kappa$ because $2\chi$ is precisely the largest response that the allowed coherent term can generate at $z=0$.  The diagonal gain still determines the certified magnitude once the threshold is crossed.

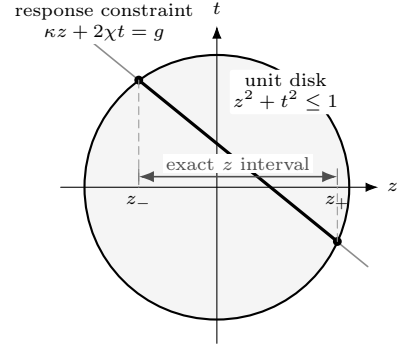
\begin{figure}[t]
\centering
\begin{tikzpicture}[scale=1.75,>=Latex,line cap=round,line join=round]
\coordinate (O) at (0,0);
\coordinate (L) at (-0.59,0.81);
\coordinate (R) at (0.91,-0.41);
\coordinate (Li) at (-0.59,0.08);
\coordinate (Ri) at (0.91,0.08);

\fill[black!4] (O) circle (1);
\draw[thick] (O) circle (1);
\draw[->] (-1.18,0)--(1.22,0) node[right,font=\scriptsize] {$z$};
\draw[->] (0,-1.18)--(0,1.22) node[above,font=\scriptsize,yshift=1pt] {$t$};

\draw[semithick,black!45] (-1.02,1.16)--(1.14,-0.60);
\draw[very thick] (L)--(R);
\fill (L) circle (0.95pt);
\fill (R) circle (0.95pt);

\draw[densely dashed,black!40] (L)--(-0.59,0);
\draw[densely dashed,black!40] (R)--(0.91,0);
\draw[thin,black!55] (-0.59,0.02)--(-0.59,0.14);
\draw[thin,black!55] (0.91,0.02)--(0.91,0.14);
\draw[semithick,<->,black!70] (Li)--(Ri)
 node[midway,above=1.2pt,font=\scriptsize,fill=white,inner sep=1pt]
 {exact $z$ interval};

\node[font=\scriptsize,below] at (-0.59,0) {$z_-$};
\node[font=\scriptsize,below] at (0.91,0) {$z_+$};

\node[font=\scriptsize,align=center,fill=white,inner sep=1.3pt,rounded corners=1pt]
 at (-0.85,1.25) {response constraint\\$\kappa z+2\chi t=g$};
\node[font=\scriptsize,align=center,fill=white,inner sep=1.3pt,rounded corners=1pt]
 at (0.52,0.72) {unit disk\\$z^2+t^2\le1$};
\end{tikzpicture}
\caption{Geometric content of Theorem~\ref{thm:bounded-interval} for $\chi>0$.  The variable $t$ is a normalized signed coherent coordinate, so the allowed coherent contribution is written as $2\chi t$ with $z^2+t^2\le1$.  The response constraint $\kappa z+2\chi t=g$ cuts a chord of the unit disk, and projecting the chord endpoints onto the $z$ axis gives the exact interval $[z_-,z_+]$.}
\label{fig:bounded-geometry}
\end{figure}

\subsection{Three information regimes}
\label{subsec:regimes}

The calibrated-family and bounded-coherence results rely on different accepted measurement information, and their conclusions should not be conflated.  Table~\ref{tab:regimes} summarizes the three inference regimes used in this paper.

\begin{table*}[t]
\caption{Inference regimes and the physical conclusion available in each.  The implementation loop is an additional audit in every row; it does not replace the listed measurement information.}
\label{tab:regimes}
\centering
\begin{ruledtabular}
\begin{tabular}{p{0.15\textwidth}p{0.25\textwidth}p{0.27\textwidth}p{0.25\textwidth}}
regime&
accepted measurement information&
population inference&
principal limitation\\
\hline
fully calibrated family&
all response rows $\bm a_j$, hence exact $A$&
closed compatible interval in Eq.~\eqref{eq:exact-interval}; exact statewise version in Eq.~\eqref{eq:absolute-interval}&
requires full characterization of each relevant binary effect\\[1.2ex]

bounded coherent response&
diagonal gain $\kappa$ and an independent bound $|c|\le\chi$&
exact worst-case interval in Eq.~\eqref{eq:zpm}; sign identified iff $g>2\chi$&
does not exploit phase information or correlations among several effects\\[1.2ex]

interval-calibrated family&
nominal rows $\widehat{\bm a}_j$ with uncertainty radii, plus response intervals&
convex outer interval from Eq.~\eqref{eq:robust-socp} below&
usually conservative; exactness depends on the uncertainty-set model
\end{tabular}
\end{ruledtabular}
\end{table*}

For several bounded-coherence readouts covered by one simultaneous confidence event, readout $j$ gives a single-readout compatible interval $[z_{j,-},z_{j,+}]$.  If all readouts probe the same active--reference state difference, any joint realization must lie in the intersection of these intervals:
\begin{equation}
 z\in\bigcap_j [z_{j,-},z_{j,+}].
\end{equation}
Thus a valid one-sided lower bound is
\begin{equation}
 z\ge\max_j z_{j,-}.
 \label{eq:multi-lower}
\end{equation}
This intersection is generally a conservative joint certificate rather than an exact joint ambiguity set, because the single-readout intervals do not enforce the common coherence coordinate $w$ across different readouts.  Using the fully calibrated response matrix $A$ is generally stronger when available, since it retains relative phase and measurement-direction information that independent scalar bounds discard.

\subsection{Benchmark values}

For the two effects in Eq.~\eqref{eq:benchmark-effects}, positivity alone gives
\begin{align}
 \chi_{{\rm pos},0}&=\sqrt{0.1125}\simeq0.3354,\\
 \chi_{{\rm pos},1}&=\sqrt{0.0325}\simeq0.1803.
 \label{eq:benchmark-pos-ceilings}
\end{align}
For the representative response $g=0.240$, the sharp lower bounds from Eq.~\eqref{eq:zpm} are approximately $-0.541$ and $-0.190$, respectively.  Positivity alone therefore does not certify a positive population change, in agreement with the explicit countermodels.

If an independent calibration instead establishes $\chi=0.002$ with $\kappa=0.60$, Eq.~\eqref{eq:zpm} gives the lower bounds
\begin{equation}
 (0.393872,0.402230,0.400558,0.407245)
 \label{eq:small-chi-bounds}
\end{equation}
for the four responses in Eq.~\eqref{eq:benchmark-responses}.  All four satisfy $g>2\chi$.  Thus the same response values and implementation-loop data can support either no population-sign conclusion or a strong positive conclusion, depending on the independently justified measurement-direction premise.

\section{Establishing the measurement-direction premise}
\label{sec:closure}

The interval formulas become population certificates only when their measurement-calibration premises are established on the same device and within the same stability window as the response experiment.  This section records several routes with different resource requirements.  None is device independent; each makes its trusted operations explicit.

\subsection{Complete calibration of one binary qubit effect}
\label{subsec:four-probe}

A general binary qubit effect has four real parameters.  In the ideal probability limit, four linearly independent trusted probe preparations therefore suffice for complete linear characterization of the event-$1$ effect.  Take
\begin{equation}
 \ket{0},\quad \ket{1},\quad
 \ket{+_x}=\frac{\ket{0}+\ket{1}}{\sqrt2},\quad
 \ket{+_y}=\frac{\ket{0}-i\ket{1}}{\sqrt2},
 \label{eq:four-probes}
\end{equation}
and denote their event probabilities by $p_0,p_1,p_x,p_y$.  In the convention of Eq.~\eqref{eq:effect-param},
\begin{align}
 \beta&=p_0,\\
 \kappa&=p_1-p_0,\\
 \operatorname{Re}c&=p_x-\frac{p_0+p_1}{2},\\
 \operatorname{Im}c&=p_y-\frac{p_0+p_1}{2}.
 \label{eq:four-probe-reconstruction}
\end{align}
Applying Eq.~\eqref{eq:four-probe-reconstruction} to every reported readout constructs the full response matrix $A$ used in Theorem~\ref{thm:exact-interval}.  This is complete calibration of the relevant binary effects, not merely an auxiliary diagonal calibration.  With finite calibration data, the same linear relations generate confidence regions for the rows of $A$, which are treated below.  Modern detector-tomography protocols may reduce statistical or computational overhead in larger systems~\cite{Barbera2025,ZambranoQMT2025,Das2026}; in the present single-qubit setting the four-parameter reconstruction is explicit.

If the diagonal entries have already been calibrated and only a bound on $|c|$ is required, phase-opposed coherent probes can isolate the off-diagonal response while canceling the common diagonal offset.  With
\[
 \ket{\pm_x}=\frac{\ket0\pm\ket1}{\sqrt2},
 \qquad
 \ket{\pm_y}=\frac{\ket0\mp i\ket1}{\sqrt2},
\]
and corresponding event probabilities $p_{\pm x}$ and $p_{\pm y}$, the convention of Eq.~\eqref{eq:effect-param} gives
\begin{align}
 \operatorname{Re}c&=\frac{p_{+x}-p_{-x}}{2},\\
 \operatorname{Im}c&=\frac{p_{+y}-p_{-y}}{2}.
 \label{eq:opposite-probes}
\end{align}
A simultaneous confidence region for the two calibration probability differences gives a valid uncertainty region for $(\operatorname{Re}c,\operatorname{Im}c)$.  If this region is represented by intervals $I_R$ and $I_I$ for the real and imaginary parts, respectively, then
\[
 \chi=\sqrt{\max_{x\in I_R}x^2+\max_{y\in I_I}y^2}
\]
is a valid scalar bound on $|c|$.  In the usual trusted-probe setting, these probability differences can be obtained from detector-tomography or direct measurement-characterization data~\cite{Zhang2020,Xu2020}.  Self-characterizing detector protocols provide a related route when the probe states are not fully trusted, but they carry different assumptions and certify the detector jointly with additional experimental structure~\cite{XuPRL2021,XuAP2021}.

\subsection{Trusted computational-basis dephasing}
\label{subsec:dephasing}

A second route removes the coherence term instead of reconstructing it.  Let
\begin{equation}
 \mathcal D_Z(\rho)=
 \ketbra{0}{0}\rho\ketbra{0}{0}
 +\ketbra{1}{1}\rho\ketbra{1}{1}
 \label{eq:z-dephase}
\end{equation}
be a trusted dephasing channel in the computational basis.  Applying it to both active and reference preparations before the same readout gives, for $X_i=\rho_i-\sigma$,
\begin{equation}
 g_{ij}^{(D)}
 =
 \Tr[F_j\mathcal D_Z(X_i)]
 =
 \kappa_j z_i.
 \label{eq:dephased-response}
\end{equation}
Thus $z_i=g_{ij}^{(D)}/\kappa_j$ whenever the diagonal gain $\kappa_j>0$ has been calibrated.  Equivalently, the same channel can be implemented as
\[
 \mathcal D_Z(\rho)=\frac12(\rho+Z\rho Z)
\]
by a randomized identity/$Z$ operation, provided that the randomization and the $Z$ axis are trusted.

This route changes the experimental intervention, so it tests the population response of the dephased preparations.  It supports an interpretation of the original preparations only when the accepted model asserts that dephasing preserves the target population and does not alter hidden leakage, loss, or selection mechanisms relevant to the registered event.  Those assumptions should be stated rather than absorbed into the word ``dephasing.''

\subsection{Positivity-only calibration is usually too weak}

When no coherent probes or trusted dephasing are available, Eq.~\eqref{eq:chi-pos} supplies a universal bound on the off-diagonal response derived only from $0\le F\le\id$.  This bound is exact as an effect-positivity constraint, but it can be much larger than the true coherent sensitivity.  The benchmark in Sec.~\ref{sec:benchmark} shows that such a positivity ceiling need not certify the sign of the population change.  In the positivity-only information model, the sharp sign condition becomes
\begin{equation}
 g>2\chi_{\rm pos}(\beta,\kappa).
 \label{eq:positivity-sign}
\end{equation}
This inequality provides a useful no-extra-probe screening test: failure does not rule out a positive population change, but it shows that diagonal calibration together with effect positivity cannot certify the sign through Theorem~\ref{thm:bounded-interval}.

\subsection{Finite statistics and interval-valued calibration}
\label{subsec:robust}

Let $\widehat g_j$ estimate the response and $\widehat{\bm a}_j$ estimate the calibrated response row.  Suppose a simultaneous confidence event guarantees, for every $j$,
\begin{equation}
 |g_j-\widehat g_j|\le s_j,
 \qquad
 \normtwo{\bm a_j-\widehat{\bm a}_j}\le\epsilon_j.
 \label{eq:row-confidence}
\end{equation}
For any physical difference vector, $\normtwo{\bm q}\le1$.  Hence any true compatible $\bm q$ in this confidence event satisfies
\begin{equation}
 |\widehat{\bm a}_j^{\mathsf T}\bm q-\widehat g_j|
 \le s_j+\epsilon_j.
 \label{eq:outer-strip}
\end{equation}
A conservative physical outer set is therefore
\begin{equation}
 \mathcal Q_{\rm out}:=
 \left\{\bm q:
 \normtwo{\bm q}\le1,
 \ |\widehat{\bm a}_j^{\mathsf T}\bm q-\widehat g_j|
 \le s_j+\epsilon_j\ \forall j
 \right\}.
 \label{eq:robust-set}
\end{equation}
The corresponding conservative population interval is obtained from the pair of second-order-cone programs
\begin{equation}
 z_{\rm L}=\min_{\bm q\in\mathcal Q_{\rm out}}\ez^{\mathsf T}\bm q,
 \qquad
 z_{\rm U}=\max_{\bm q\in\mathcal Q_{\rm out}}\ez^{\mathsf T}\bm q.
 \label{eq:robust-socp}
\end{equation}
These programs are small, convex, and globally solvable~\cite{Boyd2004}.  On the stated simultaneous confidence event, the true population change lies in $[z_{\rm L},z_{\rm U}]$ whenever the underlying qubit model and calibration assumptions are valid.  If response and calibration estimates have important correlations, a sharper certificate should optimize over their joint confidence region rather than replace it by independent Euclidean row radii and response intervals.  Equation~\eqref{eq:robust-set} remains valid under the stated simultaneous event, but it may be conservative.

For Bernoulli streams, let $a$ label the individual probability estimates entering the response, calibration, or probe experiments, and let $n_a$ be the number of repetitions used for estimate $\widehat p_a$.  One explicit simultaneous construction is
\begin{equation}
 |\widehat p_a-p_a|
 \le\sqrt{\frac{\log(2/\alpha_a)}{2n_a}},
 \qquad \sum_a\alpha_a\le\alpha,
 \label{eq:hoeffding}
\end{equation}
which holds for all indexed probabilities simultaneously with probability at least $1-\alpha$, by Hoeffding's inequality and a union bound~\cite{Hoeffding1963}.  Likelihood regions, confidence sequences, or the state-function certification methods developed for partial information can replace this elementary box construction~\cite{Christandl2012,Zambrano2024,Zambrano2026}.  Probe-state preparation uncertainty can be included using
\begin{equation}
 \frac12\normone{\rho_a-\rho_a^{\rm ideal}}\le u_a
 \quad\Longrightarrow\quad
 |\Tr[F(\rho_a-\rho_a^{\rm ideal})]|\le u_a,
 \label{eq:probe-error}
\end{equation}
where the implication uses $0\le F\le\id$ and the variational characterization of the trace distance between states.

For the bounded-coherence regime, suppose a joint confidence event gives
\begin{equation}
 g\ge g_{\rm L},
 \qquad 0<\kappa\le\kappa_{\rm U},
 \qquad |c|\le\chi_{\rm U}.
 \label{eq:interval-premises}
\end{equation}
If $g_{\rm L}>2\chi_{\rm U}$, then the sign is certified, and a conservative sharp-form lower bound is
\begin{equation}
 z\ge
 \frac{\kappa_{\rm U}g_{\rm L}
 -2\chi_{\rm U}\sqrt{\kappa_{\rm U}^2+4\chi_{\rm U}^2-g_{\rm L}^2}}
 {\kappa_{\rm U}^2+4\chi_{\rm U}^2},
 \label{eq:finite-zlower}
\end{equation}
provided the square root is real.  Appendix~\ref{app:bounded-proof} proves the monotonicity needed for this substitution on the positive branch.

\subsection{Combined reporting rule}
\label{subsec:combined-reporting}

A population-response statement should report two logically distinct outputs:
\begin{enumerate}
\item an implementation certificate, such as Eqs.~\eqref{eq:margin-cert} and \eqref{eq:loop-cert}, for the registered source--readout family;
\item a population interval or lower bound derived from one declared measurement-information regime in Table~\ref{tab:regimes}.
\end{enumerate}
On a simultaneous confidence event, a positive population response with margin $\zeta$ is supported only when the implementation certificate passes and
\begin{equation}
 \min_i z_{i,\rm L}\ge\zeta,
 \label{eq:joint-acceptance}
\end{equation}
where $z_{i,\rm L}$ is computed from the fully calibrated family, the bounded-coherence formula, or the robust conic program.  Failure of the implementation certificate means that the registered substitutions have not been certified as stable within the stated tolerance.  Failure of Eq.~\eqref{eq:joint-acceptance} means that the accepted measurement information is insufficient to certify the target population response at margin $\zeta$.  These are different diagnoses and motivate different follow-up experiments or calibrations.

\section{Discussion}
\label{sec:discussion}

\subsection{Established background and added results}

The operator-span condition underlying Eq.~\eqref{eq:row-identifiability} is established in incomplete and indirect quantum estimation~\cite{DAriano2008,Carmeli2017,Keith2018}.  Likewise, off-diagonal POVM components, their quantification as measurement coherence, and their influence on readout probabilities are established features of quantum measurement characterization~\cite{Zhang2012,Cimini2019,Baek2020,Xu2020,Tang2024,Malik:2025kpb}.  These ingredients form the background for the present analysis.

The added result is their combination in a target-specific qubit readout problem.  First, the inverse problem is put into the normal form $\bm g=A\bm q$ and solved as the exact compatible population interval in Eq.~\eqref{eq:exact-interval}, whose width separates a measurement-design factor from a data-dependent physical factor.  Second, the same-record benchmark makes the ambiguity explicit at the level of registered finite data: legal effects, a single parent measurement, the same computational-basis calibration, the same six-probability table, and the same accepted finite-family loop remain compatible with all three population signs.  Third, the reduced information model with trusted diagonal gain and bounded coherent response is solved analytically by Eq.~\eqref{eq:zpm}, including the if-and-only-if threshold in Eq.~\eqref{eq:sharp-threshold}.  Together these results turn a qualitative warning about coherent readout into a quantitative hierarchy of measurement information.

The calibrated-family interval is low-dimensional linear geometry in the qubit difference ball.  Its value is transparency: the quantity
\begin{equation}
 \Nz=\normtwo{P_{\ker A}\ez}
\end{equation}
is a directly computable design diagnostic.  It measures how much of the population axis lies in the null space of the calibrated response map before any particular response is observed.  For compatible data, the factor $\sqrt{1-\normtwo{A^+\bm g}^2}$ then reports how much of that structural blind direction remains physically available at the observed response.

\subsection{Relation to tomography and convex-optimization certification}

Full state tomography reconstructs more than is needed for one population coordinate, while incomplete tomography can still identify selected functionals~\cite{ParisRehacek2004,Teo2013,Carmeli2017}.  All-optical tests of entropic and coherence uncertainty relations provide a complementary experimental example in which basis-dependent coherence, measurement probabilities, and reconstructed states are compared within a finite-dimensional tomography workflow~\cite{Ding2020CoherenceUncertainty}.  More general convex-optimization and confidence-region approaches can optimize target functionals over compatible states, imperfect measurements, and uncertainty sets; they remain the natural route for high-dimensional systems, nonlinear targets, correlated calibration uncertainty, or additional physical constraints~\cite{Keith2018,Zambrano2024,Zambrano2026}.  Equations~\eqref{eq:exact-interval} and \eqref{eq:zpm} are closed-form qubit reductions of this broader logic.  They provide exact benchmarks for numerical implementations and make explicit how the coherent-response bound $\chi$ controls the population-sign transition.

Measurement tomography addresses the complementary unknown object, namely the readout effect itself~\cite{Fiurasek2001,Lundeen2009,Feito2009}.  Recent work improves projective methods, sample complexity, and precision limits for detector characterization~\cite{Barbera2025,ZambranoQMT2025,Das2026}.  Once such a protocol supplies $A$ and a confidence set, Secs.~\ref{sec:exact-interval} and \ref{subsec:robust} convert that calibration information into a target-specific population interval or sign certificate.  When only a scalar bound on coherent response is required, coherent-probe contrasts can avoid carrying a full reconstructed effect into the final claim, although the calibration experiment still probes the relevant off-diagonal degrees of freedom.

\subsection{Relation to loop and self-consistent diagnostics}

Loop tomography and gate-set methods test whether data admit a context-independent operational description and quantify unmodeled implementation error~\cite{Jackson2015,McCormick2017,Rudinger2019,Nielsen2021}.  They are essential when source and measurement labels may drift, interact, or carry gauge freedom.  Additional assumptions or operations can go further and separate state-preparation from measurement errors in more specialized settings~\cite{Lin2021,Khan2026}.  The implementation certificate used here is deliberately finite and specific to the registered source--readout family: it checks the declared substitutions, while the readout-calibration information determines which measurement directions are available for population identification.

The benchmark shows why successful consistency diagnostics need not identify a target axis.  A common-mode measurement component can survive every registered replacement and therefore remain invisible to differential checks.  This is not a defect of loop methods; it is a mismatch between the question they answer and the stronger physical conclusion being requested.  In practical QCVV language~\cite{BlumeKohout2025,Hashim2025}, implementation adequacy and target-specific inference should be reported as separate layers.

\subsection{Relation to measurement coherence and readout mitigation}

A purely classical assignment matrix maps computational-basis populations to observed probabilities.  Such a population-only model is exact for arbitrary input states only when the relevant POVM elements are diagonal in that basis; otherwise the response also depends on off-diagonal state components~\cite{Cimini2019,Tang2024,Malik:2025kpb}.  Measurement-coherence monotones quantify this basis-dependent coherent part of a POVM~\cite{Baek2020}.  For a binary qubit POVM $\{F,\mathbb I-F\}$ with off-diagonal element $c$ in the computational basis, the $\ell_\infty$ measurement coherence is $|c|+|-c|=2|c|$.  Thus the threshold $g>2\chi$ in Theorem~\ref{thm:bounded-interval} has a direct resource meaning: the observed response must exceed the largest response that the allowed measurement coherence can generate without any population change.

Full detector characterization can reveal departures from a purely classical assignment-matrix model, and mitigation strategies that use a reconstructed POVM or an explicit coherent-response model can in principle retain this larger response space~\cite{Chen2019,Aasen2024}.  By contrast, detector-tomography-assisted classical correction and scalable assignment-matrix mitigation are exact or controlled within a population-response model, typically assuming invertible or structured classical readout noise~\cite{Maciejewski2020,Nation2021}.  Direct measurement-characterization and detector-reconstruction experiments show that the off-diagonal measurement degrees of freedom used here are experimentally meaningful and accessible~\cite{Zhang2012,Zhang2020,Xu2020}; self-characterizing detector protocols address related coherent detector structure under different trust assumptions~\cite{XuPRL2021,XuAP2021}.  Theory and numerical studies of coherent measurement noise further motivate treating these terms as part of the readout response rather than as a purely classical assignment error~\cite{Tang2024}.  The joint-measurability viewpoint used in the benchmark is standard in the POVM literature~\cite{Yu2010}.  The present use of these tools is target-specific: it asks what can be concluded about one population difference from the registered response and the declared measurement information.

\subsection{Imperfect measurements as a certification resource question}

Work on imprecise measurements has shown that small implementation deviations can weaken or invalidate conclusions about tomography, entanglement, and steering~\cite{Rosset2012,Morelli2022,Tavakoli2024,Svegborn2025,Moreno2026}.  Robust confidence-region approaches incorporate bounded measurement mismatch into certified state properties~\cite{Zambrano2026}, and repeated-measurement analyses can use physicality constraints to turn measurement-backaction ambiguities into finite compatible intervals~\cite{Chu2026}.  The present qubit formulas isolate a terminal-readout instance: the worst-case uncalibrated degree of freedom is the coherent response in the plane orthogonal to the desired population coordinate, and its decisive scale is $2\chi$.

This perspective also clarifies resource allocation.  Repeating many scalar-shift-related readouts can improve statistical precision and test implementation stability, but it does not reduce $\Nz$.  A single well-characterized readout direction that constrains the previous blind component can reduce structural ambiguity more effectively.  Conversely, when the experiment only needs a sign conclusion, establishing $\chi<g/2$ on the same confidence event can be cheaper than reconstructing every matrix element to high precision.

\subsection{Scope and limitations}
\label{subsec:limits}

The results rely on a trusted qubit Hilbert-space model.  Leakage, loss, postselection, or hidden modes can enlarge the feasible set and must be modeled separately.  The readout effects are assumed stable over the response and calibration windows; the implementation loop can detect some forms of drift, but it cannot prove stationarity on untested states or untested readout configurations.  Trusted probe preparations are required for direct effect calibration, and any dephasing operation must be trusted in the same computational basis used to define the target population.  No device-independent claim, nor any certification of contextuality, steering, or entanglement, is intended.

The exact interval in Theorem~\ref{thm:exact-interval} assumes a fully calibrated response matrix $A$ and exact consistency.  Finite-data analysis should use a joint confidence set, as in Sec.~\ref{subsec:robust}.  The Euclidean row-error model is a convenient outer approximation, but it is not universally optimal; a laboratory-specific covariance model, likelihood region, or joint calibration-response confidence region can be substantially tighter.  The bounded-coherence theorem is exact for the information domain it states.  Additional trusted phase information, multiple-readout correlations, source constraints, or dynamical assumptions can shrink the compatible interval further.

Finally, the paper establishes a theoretical same-record benchmark and an analytic target-certification framework.  It does not report an experimental data set.  Demonstrating the resource tradeoff on a platform with independently reconstructed coherent readout would be a natural next test, but such a demonstration is not required for the internal validity of the present qubit results.

\section{Conclusion}
\label{sec:conclusion}

A stable response and an identified physical target are different experimental achievements.  The finite-family loop used here audits whether a registered active--reference signal survives declared source and readout substitutions.  The measurement-direction information determines whether that signal identifies computational-basis population.

For fully calibrated binary qubit readouts, the complete population ambiguity is the closed interval in Eq.~\eqref{eq:exact-interval}.  Its width is the product of a measurement-family blind spot and the residual physical-state radius.  The exact same-record benchmark shows that the distinction is operational rather than semantic: one underlying outcome-record architecture, one basis calibration, one probability table, and one accepted loop admit positive, zero, and negative population shifts.  Under the weaker premise $|c|\le\chi$, Eq.~\eqref{eq:zpm} gives the exact minimax interval and $g>2\chi$ is the sharp sign threshold.

The appropriate experimental report is therefore a pair of outputs: one implementation certificate and one target-identification certificate.  Their separation identifies whether additional repetitions, a new measurement direction, coherent calibration probes, trusted dephasing, or a more complete uncertainty model is the scientifically relevant next resource.

\section*{Acknowledgments}

This work is supported in part by the Youth Project of Scientific Research Program of Anhui Provincial Department of Education (2025AHGXZK40234), by the Hefei Institute of Technology Talent Research Fund (2025KY41).

\appendix

\section{Physical qubit balls and operator-span relation}
\label{app:geometry}

\subsection{State-difference ball}

Every trace-zero Hermitian qubit operator can be written as
\begin{equation}
 X=
 \begin{pmatrix}
 -z & w^*\\
 w & z
 \end{pmatrix},
 \qquad
 \bm q=(z,\operatorname{Re}w,\operatorname{Im}w)^{\mathsf T}.
\end{equation}
Its eigenvalues are $\pm r$, where
\begin{equation}
 r=\sqrt{z^2+|w|^2}=\normtwo{\bm q}.
\end{equation}
If $X=\rho-\sigma$ for density operators, then $\normone{X}\le2$.  Since $\normone{X}=2r$ for the above trace-zero qubit operator, this implies $r\le1$.  Conversely, if $r\le1$, then
\begin{equation}
 \rho=\frac{\id+X}{2},
 \qquad
 \sigma=\frac{\id-X}{2}
\end{equation}
are positive, unit-trace operators and $X=\rho-\sigma$.  This proves Eq.~\eqref{eq:difference-ball} and the endpoint attainability used in Theorem~\ref{thm:exact-interval}.

For one state, write
\begin{align}
 \rho&=
 \frac{\id}{2}+
 \begin{pmatrix}
 -(z-1/2) & w^*\\
 w & z-1/2
 \end{pmatrix},\\
 \bar{\bm q}&=(z-1/2,\operatorname{Re}w,\operatorname{Im}w)^{\mathsf T}.
\end{align}
The two eigenvalues are $1/2\pm\normtwo{\bar{\bm q}}$, so positivity is equivalent to $\normtwo{\bar{\bm q}}\le1/2$.  This proves the radius used in Corollary~\ref{cor:absolute}.

\subsection{Connection with the established span criterion}

Let $\Herm_0(d)$ be the real Hilbert space of trace-zero Hermitian operators on a $d$-dimensional Hilbert space, equipped with the Hilbert--Schmidt inner product
\[
 \langle X,Y\rangle=\Tr(XY).
\]
The response map associated with effects $F_1,\ldots,F_m$ is
\begin{equation}
 \mathcal M(X)=(\Tr(F_1X),\ldots,\Tr(F_mX)).
\end{equation}
For a trace-zero target observable $T$, the value $\Tr(TX)$ is determined by $\mathcal M(X)$ for all physical state differences if and only if
\begin{equation}
 T\in\operatorname{span}_{\R}
 \left\{F_j-\frac{\Tr F_j}{d}\id\right\}_{j=1}^m.
 \label{eq:appendix-span}
\end{equation}
Indeed, membership gives an explicit real linear combination of the measured responses.  Conversely, if membership fails, let $Y$ be the nonzero component of $T$ orthogonal to the span in Eq.~\eqref{eq:appendix-span}.  Then $Y\in\ker\mathcal M$, where
\[
 \ker\mathcal M:=\{X\in\Herm_0(d):\mathcal M(X)=0\},
\]
and $\Tr(TY)=\Tr(Y^2)>0$.  A sufficiently small multiple of $Y$ is a physical state difference, for example because every sufficiently small trace-zero Hermitian perturbation can be written as the difference of two density operators.  Hence equal measured responses can coexist with different target values.  This standard argument underlies Refs.~\cite{DAriano2008,Carmeli2017,Keith2018}.

For the qubit population target
\begin{equation}
 T_z=\ketbra{1}{1}-\frac{\id}{2},
\end{equation}
Eq.~\eqref{eq:appendix-span} becomes $\ez\in\operatorname{row}(A)$ under the real coordinate representation used in Sec.~\ref{subsec:qubit-coordinates}.  Theorem~\ref{thm:exact-interval} refines this binary identifiable/nonidentifiable conclusion into a data-dependent compatible interval.

\section{Exact checks for the ambiguity benchmark}
\label{app:benchmark}

For a $2\times2$ Hermitian matrix, nonnegative diagonal entries and nonnegative determinant are sufficient for positive semidefiniteness.  The effects in Eq.~\eqref{eq:benchmark-effects} obey
\begin{align}
 \det F_0&=1801/10000,&
 \det(\id-F_0)&=801/10000,\\
 \det F_1&=3/10,&
 \det(\id-F_1)&=0,\\
 \det(F_1-F_0)&=99/10000.
 \label{eq:benchmark-determinants}
\end{align}
Thus all four parent effects in Eq.~\eqref{eq:parent-povm} are positive and sum to the identity.

For the nine states in Table~\ref{tab:benchmark-states}, positivity is equivalent to
\begin{equation}
 z(1-z)-|w|^2\ge0.
\end{equation}
The determinants for $(\rho_-,\rho_0^+,\rho_1^+)$ are
\begin{align}
 \mathcal M_+:&\quad
 (47/900,\ 31/180,\ 22667/129600),\\
 \mathcal M_0:&\quad
 (7/180,\ 7/180,\ 4787/129600),\\
 \mathcal M_-:&\quad
 (7/1800,\ 71/22500,\ 899/3240000),
 \label{eq:state-determinants}
\end{align}
all strictly positive.

The event probability is
\begin{equation}
 \Tr[F_j\rho(z,w)]
 =\beta_j+\kappa_j z+2\operatorname{Re}(c_jw).
 \label{eq:probability-substitution}
\end{equation}
Direct rational substitution yields Eq.~\eqref{eq:benchmark-probabilities}.  The endpoint residuals then follow from Eqs.~\eqref{eq:rh} and \eqref{eq:rv}.  For example,
\begin{equation}
 r_0^{\rm h}
 =
 \left|\frac{99}{200}-\frac{49}{100}\right|
 =
 \frac{1}{200},
\end{equation}
and for the left vertical edge at $\Delta=1/10$,
\begin{equation}
 x_{01}-\mathcal A_{1/10}x_{00}
 =
 \left(\frac{1}{500},-\frac{1}{500}\right),
\end{equation}
so $r_0^{\rm v}(1/10)=1/500$.  The remaining edge residuals are obtained in the same way and give Eq.~\eqref{eq:benchmark-residuals}.

For the absolute-state intervals, Eq.~\eqref{eq:benchmark-A} gives
\begin{equation}
 \ker A=\operatorname{span}\{(-3/5,1,0)^{\mathsf T}\},
 \qquad \Nz=3/\sqrt{34}.
\end{equation}
The exact intervals underlying Table~\ref{tab:benchmark-intervals} are
\begin{align}
 z^{(-)}&\in
 \left[\frac{9}{68}-\frac{3\sqrt{191}}{340},
       \frac{9}{68}+\frac{3\sqrt{191}}{340}\right],\\
 z^{(0+)}&\in
 \left[\frac{29}{68}-\frac{3\sqrt{791}}{340},
       \frac{29}{68}+\frac{3\sqrt{791}}{340}\right],\\
 z^{(1+)}&\in
 \left[\frac{353}{816}-\frac{\sqrt{118151}}{1360},
       \frac{353}{816}+\frac{\sqrt{118151}}{1360}\right].
 \label{eq:benchmark-exact-state-intervals}
\end{align}
The feasible sets for the active and reference preparations are independent Cartesian factors once the common effects and the three probability rows are fixed.  Hence each endpoint of a population-difference interval is attained by choosing the corresponding endpoints of the two relevant state intervals.  This proves Eq.~\eqref{eq:benchmark-difference-intervals}.

\section{Sharp tree-to-loop factor}
\label{app:loop}

Fix a common translation parameter $\Delta$.  Suppose the right vertical edge is omitted.  Translation preserves differences, so
\begin{align}
 x_{11}-\mathcal A_\Delta x_{10}
 &=(x_{11}-x_{01})
 +(x_{01}-\mathcal A_\Delta x_{00})\nonumber\\
 &\quad+\mathcal A_\Delta x_{00}-\mathcal A_\Delta x_{10}.
\end{align}
Taking the $\ell_\infty$ norm gives
\begin{equation}
 r_1^{\rm v}(\Delta)
 \le r_1^{\rm h}+r_0^{\rm v}(\Delta)+r_0^{\rm h}.
 \label{eq:three-factor}
\end{equation}
The other omitted-edge cases follow from the same triangle-inequality argument around the square.  Hence three residuals along a path around the loop, each bounded by $\epsilon$, imply in general only a $3\epsilon$ bound on the fourth edge.

The factor is sharp already in a one-coordinate projection.  Choose three signed edge errors of magnitude $\epsilon$ with the same orientation; their sum gives an omitted residual $3\epsilon$.  For example, at $\epsilon=0.01$ and $\Delta=0.10$, the scalar endpoints
\begin{equation}
 x_{00}=0.49,
 \quad x_{10}=0.50,
 \quad x_{01}=0.58,
 \quad x_{11}=0.57
\end{equation}
give three residuals $0.01$ and the fourth residual $0.03$.  The same construction can be embedded into one coordinate of the endpoint pairs while holding the other coordinate fixed.  This sharp tree-to-loop factor is an implementation-audit fact; it does not affect the measurement-direction rank or population identifiability.

\section{Further properties of the bounded-coherence interval}
\label{app:bounded-proof}

\subsection{Support-function derivation}

For fixed $z$, Eq.~\eqref{eq:difference-ball} implies $|w|\le\sqrt{1-z^2}$.  Optimizing over the phase and magnitude of $c$ gives
\begin{equation}
 \max_{|c|\le\chi,\,|w|\le\sqrt{1-z^2}}
 2\operatorname{Re}(cw)=2\chi\sqrt{1-z^2}.
\end{equation}
Thus the lower endpoint is the smallest $z$ for which the observed response can be generated at all, namely the smallest solution of
\begin{equation}
 g\le\kappa z+2\chi\sqrt{1-z^2}.
 \label{eq:support-bound}
\end{equation}
At the lower endpoint, equality holds.  Setting $b=2\chi$ and $R^2=\kappa^2+b^2$, the line--circle intersection
\begin{equation}
 \kappa z+bt=g,
 \qquad z^2+t^2=1
\end{equation}
gives
\begin{equation}
 z=\frac{\kappa g\pm b\sqrt{R^2-g^2}}{R^2},
\end{equation}
which is Eq.~\eqref{eq:zpm}.

For any $z\le0$, the largest response allowed by coherence is bounded by
\begin{equation}
 \kappa z+2\chi\sqrt{1-z^2}\le2\chi.
\end{equation}
Hence $g>2\chi$ excludes all realizations with $z\le0$.  Conversely, when $0\le g\le2\chi$, the choice $z=0$ is feasible by using a coherent contribution of size $g$.  This proves Corollary~\ref{cor:sharp-sign} without squaring.

The same argument also explains why a differential bound on $c_1-c_0$ cannot replace an absolute bound on the coherent response of the readout used for the sign certificate.  A differential bound may show that two implementations have nearly the same coherent component, but it does not bound the common coherent component itself; that common component can still generate a response at $z=0$.

\subsection{Monotonicity on the positive branch}

Assume $g>2\chi$ and $g<R$.  Write $b=2\chi$ and define angles
\begin{align}
 \cos\alpha&=\frac{g}{R},&
 \cos\phi&=\frac{\kappa}{R},\\
 \sin\phi&=\frac{b}{R},&
 R^2&=\kappa^2+b^2.
\end{align}
Then the lower endpoint can be written as
\begin{equation}
 z_-=\cos(\alpha+\phi)>0.
\end{equation}
For fixed $g$, increasing $b$ increases both $\alpha$ and $\phi$, and therefore decreases $z_-$.  Increasing $\kappa$ increases $\alpha$ and decreases $\phi$, but on the positive branch $g>b$ the net derivative of $\alpha+\phi$ with respect to $\kappa$ is still positive.  Equivalently, direct differentiation of Eq.~\eqref{eq:zpm} gives
\begin{equation}
 \frac{\partial z_-}{\partial g}>0,
 \qquad
 \frac{\partial z_-}{\partial\kappa}<0,
 \qquad
 \frac{\partial z_-}{\partial\chi}<0
 \label{eq:monotonicity}
\end{equation}
throughout the interior region $g>2\chi$ and $g<R$; boundary cases follow by continuity.  Consequently, substituting $(g_{\rm L},\kappa_{\rm U},\chi_{\rm U})$ in Eq.~\eqref{eq:zpm} gives a conservative lower bound whenever $g_{\rm L}>2\chi_{\rm U}$ and the parameter box is compatible.

A simpler but weaker inequality follows from
$2\operatorname{Re}(cw)\le2\chi$:
\begin{equation}
 z\ge\frac{g-2\chi}{\kappa}.
 \label{eq:linear-weak-bound}
\end{equation}
It has the same sign threshold but does not exploit the state-difference disk.

\section{Finite-sample bookkeeping}
\label{app:statistics}

This appendix gives one explicit allocation; the main results do not depend on Hoeffding intervals specifically.  Suppose $K$ Bernoulli probabilities are estimated, possibly from streams with different sample sizes.  Let $a=1,\ldots,K$ label these probability estimates, and choose $\alpha_a>0$ with $\sum_a\alpha_a\le\alpha$.  Equation~\eqref{eq:hoeffding} and a union bound give simultaneous coverage at least $1-\alpha$.  Arbitrary correlation between two reported readout bits computed from the same shot does not invalidate these marginal Bernoulli intervals, because no product factorization between the reported bits is used.

If $[\ell_{ij}^{+},u_{ij}^{+}]$ and $[\ell_j^{-},u_j^{-}]$ are simultaneous confidence intervals for the active and reference endpoint probabilities, then
\begin{equation}
 g_{ij}\ge \ell_{ij}^{+}-u_j^{-}.
\end{equation}
The horizontal residual bounds in Eq.~\eqref{eq:loop-cert} are obtained by maximizing absolute affine expressions over the same confidence box, so their extrema occur at box endpoints.  For the vertical part of Eq.~\eqref{eq:loop-cert}, one first fixes a common $\Delta\in\calD$, maximizes the corresponding absolute affine residuals over the confidence box, and then performs the prescribed one-dimensional minimization over $\Delta\in\calD$.  The maximization over probabilities again occurs at box endpoints; the optimization over $\Delta$ is a convex piecewise-linear problem.  The same probability estimate must not be counted twice as independent evidence when it appears in several edges.

For coherent calibration, the linear combinations in Eqs.~\eqref{eq:four-probe-reconstruction} and \eqref{eq:opposite-probes} should be evaluated on one joint event.  If the response, diagonal calibration, and coherent probes use failure budgets $\alpha_{\rm resp}$, $\alpha_{\rm diag}$, and $\alpha_{\rm coh}$, respectively, then all reported implementation and population statements hold jointly with probability at least
\begin{equation}
 1-\alpha_{\rm resp}-\alpha_{\rm diag}-\alpha_{\rm coh}.
\end{equation}
More efficient simultaneous regions can be substituted without changing the separation between the implementation certificate and the target-identification certificate.

\bibliography{references}

\end{document}